\documentclass{jstpip}
\usepackage[dvips]{graphicx}
\usepackage{amsmath,amsthm,amscd,amssymb}
\usepackage{latexsym}
\usepackage[dvips]{color}
\usepackage{url}
\newtheorem{theorem}{Theorem}
\newtheorem{example}{Example}
\def\Real{\mathbb{R}}
\usepackage{color}

\begin{document}
\title{
Pooling Design and Bias Correction in DNA Library Screening
}
\begin{authors}
 \vspace*{4mm}
 \author{Takafumi Kanamori}{Nagoya University}{Furocho, Chikusaku, Nagoya 464-8603, Japan}
 {kanamori@is.nagoya-u.ac.jp}\vspace*{3mm}
 \author{Hiroaki Uehara}{Nagoya University}{Furocho, Chikusaku, Nagoya 464-8603, Japan}
 {uehara@jim.math.cm.is.nagoya-u.ac.jp}\vspace*{3mm}
 \author{Masakazu Jimbo}{Nagoya University}{Furocho, Chikusaku, Nagoya 464-8603, Japan}
 {jimbo@is.nagoya-u.ac.jp}
\end{authors}

\begin{abstract}
 We study the group test for DNA library screening based on probabilistic approach. 
 Group test is a method of detecting a few positive items from among a large number of 
 items, and has wide range of applications. 
 In DNA library screening, positive item corresponds to the clone having a specified DNA
 segment, and it is necessary to identify and isolate the positive clones for compiling
 the libraries. 
 In the group test, a group of items, called {\it pool}, is assayed in a lump in order to
 save the cost of testing, and positive items are detected based on the observation from
 each pool. 
 It is known that the design of grouping, that is, 
 pooling design is important to 
 achieve accurate detection. 
 In the probabilistic approach, 
 positive clones are picked up based on the posterior probability.
 Naive methods of computing the posterior, however, involves exponentially many sums, 
 and thus we need a device. 
 Loopy belief propagation (loopy BP) algorithm is one of popular methods to obtain
 approximate posterior probability efficiently. 
 There are some works investigating the relation between the accuracy of the loopy BP and
 the pooling design. Based on these works, 
 we develop pooling design with small estimation bias of posterior probability, and we
 show that the balanced incomplete block design (BIBD) has nice property for our purpose. 
 Some numerical experiments show that the bias correction under the BIBD is useful to
 improve the estimation accuracy. 

\end{abstract}

\section{Introduction}
\label{sec:Introduction}
We study the group test based on a probabilistic approach. Group test is a method of
detecting positive items out of a set of a large number of items, and has wide range of
applications such as blood test or DNA library screening. 

In the context of DNA library screening, our purpose is to identify clones having a
specified DNA fragment from among a collection of DNA segments. Each DNA segment is called
clones. The clone with a specified segment is referred to as {\it positive} clone,
otherwise {\it negative} clone. For large libraries, it is impractical to screen each
clone individually, instead a group of clones, called {\it pool}, is assayed in a
lump. This is said to be group test or pooling experiment. When a pool gives positive
result, the pool contains at least one positive clone, and otherwise all clones are
negative. 
A number of pools are prepared, and outcomes from all pools are assembled to identify
positive clones. 

There are mainly two categories of group test; one is adaptive, and the other is
non-adaptive. In adaptive strategy, the pool is sequentially prepared and the test is
conducted based on the information of previous outcomes. By repeating the test procedure,
we can narrow down the set of positive clones. 
In non-adaptive testing, 
we prepare all pools to be tested before conducting the group test. 
The positive clones are detected based on the outcome of each pool. That is, the 
grouping of clones does not depend on the result of previous testing. 
When the group test for each pool is performed by distinct experimenters, non-adaptive 
method may not be time-consuming compared to adaptive one. 
In this article, we focus on non-adaptive testing. 

In group testing, we have two kind of detecting procedure; 
one is combinatorial and the other is probabilistic. 
In combinatorial group testing, the main issue is to construct the design
of grouping or {\it pooling design} to reduce the number of testing 
without missing the positive clones. 
Combinatorial group testing has been studied by many authors
\citep{d-z.99:_combin_group_testin_and_its_applic_ed,ngo00:_survey_combin_group_testin_algor,
wu04:_molec_biolog_and_poolin_desig}. 
In combinatorial approach, it is often assumed that the maximum number
of positive clones is known and that there is no observation errors or noisy
measurements. 
On the other hand, in probabilistic approach the prior probability for the state of clones
is assumed, and 
posterior probability such that each clone is positive is computed based on the
observation of each pool 
\citep{knill96:_inter_of_poolin_exper_using,bruno95:_effic_poolin_desig_for_librar_screen,
mezard07:_group_testin_with_random_pools,ueharaar:_posit_detec_code_and_its}. 
The main issue is to develop efficient algorithm to compute the posterior probability,
since using naive Bayes formula is computationally demanding.
\cite{knill96:_inter_of_poolin_exper_using} and \cite{ueharaar:_posit_detec_code_and_its} 
have proposed a probabilistic algorithm. 
\citet{knill96:_inter_of_poolin_exper_using} have used 
the Markov Chain Monte Carlo (MCMC) method to obtain the marginal posterior probability,
and \citet{ueharaar:_posit_detec_code_and_its} have exploited the 
{\it loopy belief propagation} (BP) algorithm 
\citep{pearl88:_probab_reason_in_intel_system,mackay99:_good_error_correc_codes_based}
to compute approximate probability. 

Non-adaptive group test with probabilistic approach will be one of the most practical
methods to detect positive clones from among large DNA library. Even in probabilistic
approach, the pooling design is significant to achieve highly accurate estimation of
posterior probability. 
In loopy BP algorithm for the low density parity check (LDPC) coding 
\citep{mackay99:_good_error_correc_codes_based,richardson01:_desig_of_capac_approac_irreg}, 
it has been revealed that the coding design
is closely related to the decoding error of the transmitted code. 
Likewise, the pooling design with some nice property will provide 
accurate estimator of the posterior probability as experimentally shown by  
\citet{ueharaar:_posit_detec_code_and_its}. 
In coding theory, \citet{ikeda04:_infor_geomet_of_turbo_and} have analyzed 
the relation between the coding design and the bias of the estimated posterior
probability. 
We apply their result to improve the accuracy of the group testing. 

The outline of the paper is as follows. In Section
\ref{sec:Prelininaries_DNA_lib_screening} probabilistic description of group testing for 
DNA library screening is presented. In Section \ref{sec:BeliefPropagationAlgorithm} we
introduce loopy belief propagation algorithm, and in Section \ref{sec:Bias_of_BP} we show
the bias the estimated posterior probability according to
\cite{ikeda04:_infor_geomet_of_turbo_and}. 
In Section \ref{sec:PoolingDesign_Bias}, we construct a pooling design resulting in a
small bias. Numerical experiments are presented in Section
\ref{sec:Numerical_Experiments}. Section \ref{sec:ConcludingRemarks} is devoted to
concluding remarks.

\section{Preliminaries of DNA library screening}
\label{sec:Prelininaries_DNA_lib_screening}
On DNA library screening, our purpose is to identify the positive clones out of 
a large DNA library. 
Let $X_i$ be the random variable which stands for the label of the clone $i$ for
$i=1,\ldots, n$, that is, $X_i=1$ for positive and $X_i=0$ for negative. 
The labels of all clones are denoted as the vector $X=(X_1,\ldots, X_n)$. 
We assume that the random variables $X_i,\,i=1,\ldots,n$ are independent. 
The probability such that $X=x\in\{0,1\}^n$ for $x=(x_1,\ldots,x_n)$ is denoted by $p(x)$
or $p(X=x)$. Then, 
the probability $p(x)$ is represented by the factorization of marginal probabilities, that is 
\begin{align*}
p(x)=p_1(x_1)\times\cdots\times p_n(x_n). 
\end{align*}
Since the marginal distribution $p_i(x_i)$ over $\{0,1\}$ is written as the form of 
exponential model $p_i(x_i)\propto\exp\{h_i x_i\}$, the joint probability $p(x)$ is given
as 
\begin{align*}
p(x)=\exp\{h^\top x-\psi_0(h)\},\quad  x\in\{0,1\}^n
\end{align*}
with $h=(h_1,\ldots,h_n)\in\Real^n$, where $\psi_0(h)$ is the normalization factor called
the cumulant generating function. 

In the group test a number of clones are set in a pool and the experiment is conducted
to detect if a positive close is included in the pool. Here the pool is identified by 
a subset of $\{1,\ldots,n\}$, and 
the clone $i$ is included in the pool $r$ if and only if $i\in r$ holds. 
For the pool $r\subset\{1,\ldots,n\}$, let $Z_r$ be the random variable defined by 
\begin{align}
 \label{eqn:def-Z_r}
 Z_r=
 \begin{cases}
  1& ^\exists i\in r,\ X_i=1,\\
  0& \text{otherwise}. 
 \end{cases}
\end{align}
Hence if $Z_r=1$, there is a positive clone in the pool $r$. 
Note that $Z_r$ is also represented as 
\begin{align*}
 Z_r=\max_{i\in r} X_i=1-\prod_{i\in r}(1-X_i). 
\end{align*}
In practice, $Z_r$ is not directly observed. 
The observation of the pool $r$ is usually represented by four levels such that
\begin{align*}
 S_r=
 \begin{cases}
  0 & \text{if the pool $r$ is negative},\\
  1 & \text{if the pool $r$ is weak positive},\\
  2 & \text{if the pool $r$ is medium positive},\\
  3 & \text{if the pool $r$ is strong positive}. 
 \end{cases}
\end{align*}
The response of the experiment is measured by using a fluorescence sign, and 
it is experimentally-confirmed that the conditional probability of $S_r$ given 
$X_i, (i\in r)$ only depends on $Z_r$, not the number of $i\in r$ such that $X_i=1$. 
We assume that the conditional probability of $S_r$ given $Z_r$ is the same for 
all pools. Then, the conditional probability of $S_r=s_r$ given $Z_r=z_r$ is denoted as  
$p(S_r=s_r|Z_r=z_r)$ or $p(s_r|z_r)$. 
In practice $p(S_r=0|Z_r=0)$ and $p(S_r=3|Z_r=1)$ will take larger value than others. 
In the group test usually we prepare a number of pools. Let
$\mathcal{G}=\{r_1,\ldots,r_m\}$ be 
the set of pools used in the group test. Then for each pool $r\in \mathcal{G}$
the observation $s_r\in\{0,1,2,3\}$ is obtained. 
An example of a pooling design is shown in Figure \ref{fig:pooling_design}. 
\begin{figure}[tp]
 \begin{center}
  \includegraphics[scale=0.5]{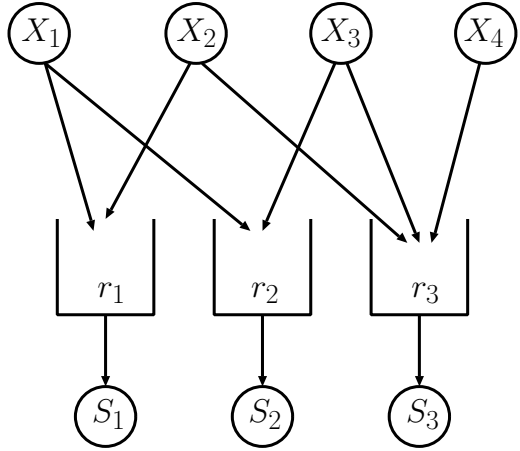}
  \caption{An example of a pooling design. 
  $\mathcal{G}$ is given as $\{\{1,2\},\,\{1,3\},\,\{2,3,4\}\}$. }
  \label{fig:pooling_design}
 \end{center}
\end{figure}

The problem considered in the paper is to infer the label of clones based on the
observation from each pool. 
More precisely, we want to pick up only positive clones out of all clones. 
As a probabilistic approach, the method of maximum a posteriori (MAP) estimate is useful
to detect the positive clones. 
Let $S=(S_1,\ldots,S_m)$ be the random variable for the observation from all pools, and 
\[
p(X=x|S=s)=p(x|s)
\]
be the posterior probability of $X=x=(x_1,\ldots,x_n)$ given $S=s=(s_1,\ldots,s_m)$, 
where $m$ is the number of pools. 
The label pattern $x\in\{0,1\}^n$ maximizing the posterior $p(x|s)$ will provide the set 
of clones which are likely to be positive. 

We represent the posterior $p(x|s)$ by $p(x)$ and $p(s|x)$. 
Using the Bayes formula, we can represent the posterior probability $p(x|s)$ as 
\begin{align*}
 p(x|s)\propto p(s|x)p(x). 
\end{align*}
For the distinct pools $r,r'\in\mathcal{G}$, the observations $s_{r}$ and $s_{r'}$ 
are conditionally independent for given $x$. Hence the probability $p(s|x)$ is decomposed
into the conditional probabilities of $r\in\mathcal{G}$, and then we have 
\begin{align*}
 p(s|x)=\prod_{r\in\mathcal{G}}p(s_r|x). 
\end{align*}
For each observation $s\in\{0,1,2,3\}$, the conditional probability $p(S_r=s|x)$ is
written as
\begin{align*}
 p(S_r=s|x)\propto \exp\{c(s,x)\}
\end{align*}
as the function of $s$, where $c(s,x)$ is a real-valued function. When we compute the
posterior probability of $p(x|s)$, the observations $s_r, r\in\mathcal{G}$ are regarded as
constants, and thus $c(s_r,x)$ is written as $c_r(x)$ as the function of the label pattern
$x\in\{0,1\}^n$. Note that $c_r(x)$ depends only on $z_r$ which is a realized value of $Z_r$ 
defined in \eqref{eqn:def-Z_r}. 
Then, the posterior probability $p(x|s)$ is given as
\begin{align}
 p(x|s)\propto\exp\big\{ h^\top x+\sum_{r\in\mathcal{G}}c_r(x)\big\}. 
 \label{eqn:posterior}
\end{align}
Suppose that the parameter $h$ and the functions $c_r(x),\, r\in\mathcal{G}$ are known or
these are estimated with satisfactory accuracy. 

In general the maximization of $p(x|s)$ in \eqref{eqn:posterior} over $x\in\{0,1\}^n$ is
computationally hard unless the set of pools $\mathcal{G}$ has some special property 
\citep{pearl88:_probab_reason_in_intel_system,cowell07:_probab_networ_and_exper_system}. 
Thus, we take another approach.  
The marginal probability of $x_i$ for $p(x|s)$ is denoted as $p_i(x_i|s)$, that is 
\begin{align}
 p_i(x_i|s)=\sum_{x'\in\{0,1\}^n\,:\,x'_i=x_i} p(x'|s)
 \label{eqn:marginal-posterior}
\end{align}
We think that the clones having large marginal posterior $p_i(X_i=1|s)$ will be
positive. 
Using a threshold for the marginal posterior, 
we will be able to detect the set of positive clones. 
As an example Table \ref{table:marginal_posterior_example} 
shows exact marginal posterior probabilities $p_i(X_i=1|s)$. 
The pooling design in Figure \ref{fig:pooling_design} and 
the observation probability $p(S_r=s|x)$ shown in Table \ref{table:observation_prob} are
used, and the marginal probability $p(X_i=1)$ is set to $0.1$ for all clones. 
We see that the marginal posterior will be useful to detect the positive clones. 

\begin{table}
 \begin{center}
  \caption{An example of marginal posterior probabilities. 
  The pooling design in Figure \ref{fig:pooling_design} and 
  the observation probability $p(S_r=s|x)$ shown in Table \ref{table:observation_prob} are
  used, and the marginal probability $p(X_i=1)$ is set to $0.1$ for all clones. 
  }
 \label{table:marginal_posterior_example}
 \begin{tabular}{|c||c|c|c|c|} 
  \hline 
  $(s_1,s_2,s_3)\in\{0,1,2,3\}^3$ & 
  $p_1(X_1=1|s)$ & $p_2(X_2=1|s)$ & $p_3(X_3=1|s)$ & $p_4(X_4=1|s)$ \\\hline
  $(3,0,0)$           & 0.043 & 0.047 & 0.001 & 0.011 \\ 
  $(2,2,0)$           & 0.853 & 0.019 & 0.019 & 0.009 \\
  $(0,1,3)$           & 0.020 & 0.016 & 0.760 & 0.180 \\
  $(0,0,3)$           & 0.001 & 0.027 & 0.027 & 0.429 \\\hline
 \end{tabular}
 \end{center}
\end{table}

The computation of the marginal posterior is still
hard, since there are exponentially many summands in
\eqref{eqn:marginal-posterior}. Despite this, we can compute an approximate posterior 
probability by applying so-called {\it loopy belief propagation} (loopy BP) algorithm. 
The details of loopy BP is briefly introduced in Section
\ref{sec:BeliefPropagationAlgorithm}.

\section{Loopy Belief Propagation for Computation of Marginal Probability}
\label{sec:BeliefPropagationAlgorithm}
Loopy belief propagation is a method of computing an approximate marginal 
probability, which is very useful in stochastic reasoning
\citep{pearl88:_probab_reason_in_intel_system,cowell07:_probab_networ_and_exper_system}. 
Let $q(x)$ be a joint probability of high dimensional binary variable 
$x=(x_1,\ldots,x_n)\in \{0,1\}^n$. In the group test $q(x)$ corresponds to the posterior
probability $p(x|s)$. 
The computation of the marginal $q_i(x_i)$ involves exponentially many sums. 
To reduce the computational cost, we approximate the joint probability $q(x)$ by a
tractable one. 
Suppose that $q(x)$ is represented by the form of \eqref{eqn:posterior}, 
that is, 
\begin{align}
 \label{eqn:joint-prob}
 q(x)\propto\exp\big\{ h^\top x+\sum_{r\in\mathcal{G}}c_r(x)\big\}
\end{align}
and we use the model 
\begin{align}
 \label{eqn:stat_model_joint-prob}
 \bar{q}(x;\theta) \propto \exp\big\{h^\top x+\theta^\top x\} 
\end{align}
to approximate $q(x)$, 
where $\theta=(\theta_1,\ldots,\theta_n)^\top\in\Real^n$ is an $n$ dimensional column vector. 
The parameter $\theta$ is determined such that the function $\theta^\top x$ is close to
$\sum_{r\in\mathcal{G}}c_r(x)$ up to additive constant. 
Then, the marginal probability of $q(x)$ will be approximately given by 
\begin{align*}
 \bar{q}_i(x_i;\theta)\propto \exp\big\{h_ix_i+\theta_i x_i\big\},\quad i=1,\ldots,n. 
\end{align*}
As a result, 
we can obtain an approximate value of the marginal probability for $q(x)$. 
The loopy BP algorithm provides an efficient method of computing the parameter 
$\theta$. 
Suppose that $\theta$ is decomposed into the sum of parameters 
$\xi_r\in\Real^n,\,r\in\mathcal{G}$, that is, 
\begin{align*}
\theta=\sum_{r\in\mathcal{G}} \xi_r. 
\end{align*}
We suppose that the function $c_r(x)$ is approximated by $\xi_r^\top x$ 
for each pool $r\in\mathcal{G}$. 
When the parameters $\xi_r,\,r\in\mathcal{G}$ are obtained in mid-flow of 
the algorithm, we show how to update these parameters. 
Let $\zeta_r$ be defined as 
\begin{align*}
\zeta_r=\sum_{s\in\mathcal{G}}\xi_s-\xi_r=\theta-\xi_r. 
\end{align*}
When the function $c_r(x)$ is approximated by $\xi_r^\top x$, 
the probability $q(x)$ is also approximated by $\exp\{h^\top x+\zeta_r^\top x+c_r(x)\}$. 
We seek the parameter $\bar{\xi}_r$ such that 
$\exp\{h^\top x+\zeta_r^\top x+\bar{\xi}_r^{\top} x\}$ approximates 
$\exp\{h^\top x+\zeta_r^\top x+c_r(x)\}$ up to the normalization constant. 
The Kullback-Leibler divergence 
\[
{\rm KL}(p,q)=\sum_{x\in\{0,1\}^n} p(x)\log\frac{p(x)}{q(x)}
\]
is used as the discrepancy measure between two probabilities $p$ and $q$ over
$\{0,1\}^n$. 
We consider the following optimization problem: 
\begin{align*}
 &\min_{\xi_r\in\Real^n} {\rm KL}(p,q)\\
 &\phantom{\min}\text{subject to}\ \ 
 p(x)\propto \exp\{h^\top x+\zeta_r^\top x+c_r(x)\},\quad 
 q(x)\propto \exp\{h^\top x+\zeta_r^\top x+\xi_r^\top x\}. 
\end{align*}
By some calculation, we see that the above problem is represented as the following form:
\begin{align}
 \label{eqn:BP-opt}
 \begin{array}{l}
  \displaystyle
   \max_{\xi_r\in\Real^n}
   \sum_{x\in\{0,1\}^n} \exp\{h^\top x+\zeta_r^\top x+c_r(x)\} \vspace*{1mm}
   \displaystyle
   \cdot\bigg\{
   \xi_r^\top x-\sum_{i=1}^n\log(1+\exp\{h_i+\zeta_{ri}+\xi_{ri}\})
   \bigg\}. 
 \end{array}
\end{align}
There are $2^{|r|}$ summands in the function to be optimized, where $|r|$ denotes the 
number of elements in the set $r$. 
When the size of the pool $r$ is not large, the objective function in the optimization
problem above is tractable. 
The parameter $\xi_r$ is updated to $\bar{\xi}_r$ which is the optimal solution of
\eqref{eqn:BP-opt}. 
In the same way, the parameters $\xi_r, r\in\mathcal{G}$ and the sum
$\theta=\sum_{r\in\mathcal{G}}\xi_r$ are updated sequentially. 
The convergent point of $\theta$ is the output of the algorithm, and we obtain the
approximated marginal probability $\bar{q}_i(x_i;\theta), i=1,\ldots,n$. 
The loopy BP algorithm is very useful in practice, though the convergence property of the
algorithm is not theoretically guaranteed under general condition. 

In the literature of DNA library screening, the function $c_r(x)$ depends on the value of
$z_r=\max_{i\in r} x_i$, 
and thus we define
\begin{align*}
 c_r(x) =
 \begin{cases}
  c_{r1}& z_r=1,\\
  c_{r0}& z_r=0.
 \end{cases}
\end{align*}
Then, the objective function of \eqref{eqn:BP-opt} has a simple form, and the updated
parameter $\bar{\xi}_r$ in the loopy BP algorithm is explicitly obtained. 
See \cite{ueharaar:_posit_detec_code_and_its} for details.

\section{Bias of Loopy Belief Propagation}
\label{sec:Bias_of_BP}
According to \cite{ikeda04:_infor_geomet_of_turbo_and} we show the bias introduced by the
loopy BP algorithm in the general setup. 
For each pool $r\in\mathcal{G}$, let $c_r(x)$ be any real-valued function depending only on 
$x_i, i\in r$, and $q(x)$ be a probability on $\{0,1\}^n$ defined as the form of
\eqref{eqn:joint-prob}.  
Let $q_i(x_i)$ be the marginal probability of $q(x)$. 
We use the statistical model \eqref{eqn:stat_model_joint-prob} to approximate the joint
probability $q(x)$. 
Let $\bar{q}(x;\theta_0)\propto\exp\{h^\top x+\theta_0^\top x\}$ be the convergent 
joint probability computed by the loopy BP algorithm applied to $q(x)$. 
Usually the estimated marginal $\bar{q}_i(x_i;\theta_0)$ is not equal to the true marginal
probability $q_i(x_i)$, 
and the difference $q_i(x_i)-\bar{q}_i(x_i;\theta_0)$ is said to be {\it bias}. 
\citet{ikeda04:_infor_geomet_of_turbo_and} have analyzed the bias of 
the loopy BP algorithm, and obtained the asymptotic formula such that 
\begin{align}
 \label{eqn:LBP-asympt-bias}
 q_i(1)-\bar{q}_i(1;\theta_0)~=~
 \frac{1}{2}\sum_{\substack{r,s\in\mathcal{G}\\ r\neq s}}B_{rsi}\ 
 +\text{(higher order terms)}, 
\end{align}
where $B_{rsi}$ is related to a geometrical curvature of statistical model
$\bar{q}(x;\theta)$. 

To show the definition of $B_{rsi}$, we need to define the matrices 
$g_{ij}, g_{ir},\widetilde{g}_{ir}$ and the third order tensor $T$. 
Let $\bar{x}_i$ and $\bar{c}_r$ be the expectation of $x_i$ and $c_r(x)$ under
$\bar{q}(x;\theta_0)$, that is
\begin{align*}
 \bar{x}_i=\sum_{x_i}x_i\bar{q}_i(x_i;\theta_0),\qquad
 \bar{c}_r=\sum_{x}c_r(x)\bar{q}(x;\theta_0). 
\end{align*}
Note that the expectation $\bar{x}_i$ is equal to $q_i(1;\theta_0)$, since $x_i$ is the
binary variable. 
The matrix $g_{ij},i,j=1,\ldots,n$ is the Fisher information matrix of the model
$\bar{q}(x;\theta)$ at $\theta=\theta_0$,  
\begin{align*}
 g_{ij}
 =
 \sum_{x_i,x_j} (x_i-\bar{x}_i) (x_j-\bar{x}_j)
 \bar{q}(x;\theta_0) 
 =
 \delta_{ij}\sum_{x_i} (x_i-\bar{x}_i)^2 \bar{q}_i(x_i;\theta_0)
 = \delta_{ij} \bar{x}_i(1-\bar{x}_i), 
\end{align*}
where $\delta_{ij}$ is the Kronecker's delta function such that 
$\delta_{ij}=1$ for $i=j$ and otherwise $\delta_{ij}=0$. 
Likewise the matrix $g_{ir}$ for $i=1,\ldots,n,\ r\in\mathcal{G}$ is defined by
\begin{align*}
 g_{ir}=\sum_{x}(x_i-\bar{x}_i)(c_r(x)-\bar{c}_r) \bar{q}(x;\theta_0), 
\end{align*}
and let $\widetilde{g}_{ir}$ be $\widetilde{g}_{ir}=g_{ir}/g_{ii}$. Moreover let the third
tensor $T$ be 
\begin{align*}
 T_{ijk}&=\sum_{x}(x_i-\bar{x}_i)(x_j-\bar{x}_j)(x_k-\bar{x}_k)\bar{q}(x;\theta_0),
 \quad i,j,k=1,\ldots,n,\\
 T_{ijr}&=\sum_{x}(x_i-\bar{x}_i)(x_j-\bar{x}_j)(c_r(x)-\bar{c}_r)\bar{q}(x;\theta_0), 
 \quad i,j=1,\ldots,n,\ r\in\mathcal{G},\\
 T_{irs}&=\sum_{x}(x_i-\bar{x}_i)(c_r(x)-\bar{c}_r)(c_s(x)-\bar{c}_s)\bar{q}(x;\theta_0), 
 \quad i=1,\ldots,n,\ r,s\in\mathcal{G}. 
\end{align*}
Then $B_{irs}$ is defined as
\begin{align}
 \label{eqn:def_Btensor}
 B_{rsi}
 =-T_{irs}-\sum_{j,k=1}^n T_{ijk}\widetilde{g}_{jr}\widetilde{g}_{ks}
 +\sum_{j=1}^{n}(T_{ijr}\widetilde{g}_{js}+T_{ijs}\widetilde{g}_{jr}),\quad
 i=1,\ldots,n,\ \ r,s\in\mathcal{G}. 
\end{align}
Once we obtain the approximate joint probability $\bar{q}(x;\theta_0)$, 
we can compute $B_{rsi}$ without knowing the target probability $q(x)$. 
Thus, according to \eqref{eqn:LBP-asympt-bias} the bias is corrected by adding
$\sum_{r\neq s}B_{rsi}/2$ to $\bar{q}_i(1;\theta_0)$.

\section{Relation between Pooling Design and Bias of Loopy BP Algorithm}
\label{sec:PoolingDesign_Bias}
We show some properties of $B_{rsi}$ defined in \eqref{eqn:def_Btensor}. 
Let $c_r(x)$ be the function on $\{0,1\}^n$ depending only on $x_i, i\in r$, 
then $c_r(x)$ is represented as the form of 
\begin{align}
 c_r(x)~=~h_r+\sum_{\ell}b_{r\ell}\prod_{i\in r}(x_i-a_{r\ell i}),\quad 
 h_r,\ b_{r\ell}\in\Real,\ a_{r\ell i}\in[0,1]. 
 \label{eqn:c_function_canonical_form}
\end{align}
This fact is shown below. 
Let $r=\{i_1,\ldots,i_{|r|}\}\subset\{1,\ldots,n\}$, and $\bar{x}$ be 
$\bar{x}=(\bar{x}_1,\ldots,\bar{x}_{|r|})\in \{0,1\}^{|r|}$, then 
we define $c_r(x)$ by
\begin{align}
 c_r(x)
 =
 \sum_{\bar{x}\in\{0,1\}^{|r|}}\ b_{\bar{x}}\prod_{j=1}^{|r|}
 \big(1-(x_{i_j}-\bar{x}_{j})^2\big)
 =
 \sum_{\bar{x}\in\{0,1\}^{|r|}}
 \bigg[b_{\bar{x}}\prod_{k=1}^{|r|}(2\bar{x}_{k}-1)\bigg]
 \prod_{j=1}^{|r|}\big(x_{i_j}-(1-\bar{x}_{j})\big). 
 \label{eqn:cr-other-expression}
\end{align}
The function $c_r(x)$ has the form of \eqref{eqn:c_function_canonical_form}, and  
the variable $x\in\{0,1\}^n$ satisfying $x_{i_j}=\bar{x}_j$ for $j=1,\ldots,|r|$ is 
mapped to $b_{\bar{x}}\in\Real$. 
By varying $b_{\bar{x}}$ any function over $\{0,1\}^n$ can be represented by the form
above. Though the parameter $a_{r\ell i}$ in \eqref{eqn:c_function_canonical_form} can be
restricted to the binary set $\{0,1\}$, we allow the mild condition $a_{r\ell i}\in[0,1]$
for convenience. 

\begin{example}
 \label{example:c_r_grouptest_LDPC}
For the group test 
\[
c_r(x)=\rho_r\cdot \max_{i\in r}x_i=\rho_r\{1-(-1)^{|r|}\prod_{i\in r}(x_i-1)\}
\]
 is used. 
 For the low density parity check (LDPC) codes the function 
\[
c_r(x)=\rho_r\cdot\prod_{i\in r}(1-2x_i)=\rho_r(-2)^{|r|}\prod_{i\in r}(x_i-1/2)
\]
is exploited. In the above, the coefficient $\rho_r$ determines the intensity contributed
from the pool $r\in\mathcal{G}$. 
\end{example}

First, we show the condition that the bias vanishes. 
\begin{theorem}
 \label{theorem:packing-optimal}
 Let $c_r$ be real-valued function over $\{0,1\}^n$ depending only on the variables 
 $x_i,\,i\in r$.  
 Let $r,s$ be distinct subsets of $\{1,\ldots,n\}$. 
 Then, for any functions $c_r(x),\,c_s(x)$ and any $i=1,\ldots,n$, 
 $B_{rsi}$ vanishes if $|r\cap s|\leq 1$. 
\end{theorem}
Theorem \ref{theorem:packing-optimal} is a direct conclusion of Theorem 7 in
\cite{ikeda04:_infor_geomet_of_turbo_and}. 
The proof is deferred to appendix \ref{appendix:B_calc} to show the explicit form of
$B_{rsi}$. 
Let the packing design be the family of sets $\mathcal{G}$ satisfying  $|r\cap s|\leq 1$
for any $r,s\in\mathcal{G}$, then Theorem \ref{theorem:packing-optimal} denotes that 
for the packing design the dominant bias term of loopy BP algorithm vanishes. 
The packing design is used in the design of group test
\citep{ueharaar:_posit_detec_code_and_its} and also in the LDPC code 
\citep{mackay99:_good_error_correc_codes_based}. 
It is numerically shown that the accuracy of approximate probability is superior to 
other designs with $|r\cap s|\geq 2$. 
In coding theory, lots of designs of low density parity check
(LDPC) code have been intensively studied, and the packing design is known as good
error-correcting code
\citep{ikeda04:_infor_geomet_of_turbo_and,mackay99:_good_error_correc_codes_based}. 
In Theorem \ref{theorem:packing-optimal} these results are extended to any function 
$c_r, r\in\mathcal{G}$. 

We consider the bias term $B_{rsi}$ for $|r\cap s|\geq 2$. 
\begin{theorem}
\label{theorem:bias-upper-bound}
 Let $c_r$ and $c_s$ be functions with the form of \eqref{eqn:c_function_canonical_form}, 
 and suppose that there exists a constant $C$ such that the coefficients 
 $b_{r\ell}, b_{s\ell}$  satisfy 
 \begin{align*}
  \sum_{\ell,\ell'}\big|b_{r\ell} b_{s\ell'}\big| \leq C.
 \end{align*}
 Let $\bar{x}_i$ be the expectation of $x_i$ under the probability 
 $\bar{q}(x;\theta_0)\propto \exp\{h^\top x+\theta_0^\top x\}$ and 
 $\delta$ be a real number
 satisfying 
\begin{align*}
 0<|\bar{x}_i-a_{r\ell i}| \leq \delta <1,\quad
 0<|\bar{x}_i-a_{s\ell i}| \leq \delta <1, 
\end{align*}
for any $i,\ell,r,s$. 
Then, the intensity of $B_{rsi}$ 
is bounded above as follows:
\begin{align}
 |B_{rsi}| 
 \leq 
 C\cdot\frac{\delta^{|r|+|s|-2}}{2}
 \left(1+\frac{1}{4\delta^2}\right)^{|r\cap s|}
\label{eqn:B_upper-bound}
\end{align}
\end{theorem}
The proof is shown in appendix \ref{appendix:Upperbound_Brsi}. It is easy to see the
right-hand of  
\eqref{eqn:B_upper-bound} is increasing function of $\delta>0$. 

\begin{example}
 The bias term in the group test is shown. The function $c_r$ is defined as 
 $\rho_r(1-\prod_{i\in r}(1-x_i))$ as shown in Example \ref{example:c_r_grouptest_LDPC}.
 Suppose $\bar{x}_i=\bar{x}$ holds for all $i=1,\ldots,n$. 
 Then the bias term $B_{rsi}$ for $i\in r\setminus s$ is
 given as   
 \begin{align*}
  |B_{rsi}|
  &~=~
  |\rho_r\rho_s|
  \bar{x}(1-\bar{x})(1-\bar{x})^{|r|+|s|-1}
  \left\{
  \left(1+\frac{\bar{x}}{1-\bar{x}}\right)^{|r\cap s|}-1-|r\cap s|\frac{\bar{x}}{1-\bar{x}}
  \right\}\\
  &~\leq~
  |\rho_r\rho_s|
  \frac{(1-\bar{x})^{|r|+|s|-2}}{2}\left(1+\frac{1}{4(1-\bar{x})^2} \right)^{|r\cap s|}. 
 \end{align*}
 The bias $B_{rsi}$ for $i\in r\cap s$ is also computed in the same way. 
 It is verified that $|B_{rsi}|$ vanishes for $|r\cap s|\leq 1$. 
 When $|r|$ and $|s|$ are fixed, minimization of $|r\cap s|$ will contribute
 to the reduction of the bias. 
\end{example}
\begin{example}
 Let $\bar{x}_i=\bar{x}$ for $i=1,\ldots,n$. 
 For the LDPC, the function 
 $c_r(x)=\rho_r (-2)^{|r|}\prod_{i\in r}(x_i-1/2)$ is used. 
 Then, $|B_{rsi}|$ for $i\in r\setminus s$ is given as 
 \begin{align*}
   |B_{rsi}|
  &=
  2|\rho_r\rho_s|\bar{x}(1-\bar{x})|2\bar{x}-1|^{|r|+|s|-1}\left\{
  \left(1+\frac{\bar{x}(1-\bar{x})}{(\bar{x}-1/2)^2}\right)^{|r\cap s|}-1
  -|r\cap s|\frac{\bar{x}(1-\bar{x})}{(\bar{x}-1/2)^2}\right\}\\
  &\leq
  |\rho_r\rho_s|\frac{|2\bar{x}-1|^{|r|+|s|-2}}{2}
  \left(1+\frac{1}{(2\bar{x}-1)^2}\right)^{|r\cap s|}, 
 \end{align*}
 It is verified that $|B_{rsi}|$ vanishes for $|r\cap s|\leq 1$. 
 When $\bar{x}\neq 1/2$, the bias $|B_{rsi}|$ is increasing in 
 $|r\cap s|$ when the size of pools $|r|$ is fixed.
 Thus minimization of $|r\cap s|$ is important to reduce the bias. 
\end{example}

The dominant bias is represented as the sum of $B_{rsi}$. 
We assume that the constants $\delta$ and $C$ in Theorem \ref{theorem:bias-upper-bound}
are also upper bounds for any pair of $r,s\in\mathcal{G}$. 
Suppose that the size of subset is fixed, i.e. $|r|=d$, and let $m=|\mathcal{G}|$. 
Then, an upper bound of the bias is given as 
\begin{align*}
\bigg|\frac{1}{2}\sum_{r\neq s} B_{rsi}\bigg|
\leq 
\frac{Cm(m-1)}{2}\cdot\delta^{2d-2}
\left(1+\frac{1}{4\delta^2}\right)^{\displaystyle\max_{r,s:r\neq s}|r\cap s|}. 
\end{align*}
Suppose that $C$ does not significantly depend on the pooling design. 
Then, the pooling design minimizing $\max_{r,s:r\neq s}|r\cap s|$ will lead a small
estimation bias when we use loopy BP algorithm to compute the approximate posterior
probability. 
In the group test $C$ is almost independent of the pooling design, when the size of the
pool, $|r|$, is fixed. Indeed we can choose $C=\max_{r,s}|\rho_r\rho_s|$, where $\rho_r$ 
is not significantly depend on the pooling design. 
In terms of the minimization of $\max_{r\neq s}\,|r\cap s|$,
we have the following theorem. 
\begin{theorem}
 \label{theorem:optimal_pooling_design}
 For fixed integers $m, n$ and $d$ we consider the optimization problem
 \begin{align}
  \label{eqn:min-max-overlap}
  &\min_{\mathcal{G}}\ \max_{r,s\in\mathcal{G}:r\neq s}\ |r\cap s|,\quad
  \text{\rm subject to}\   |r|=d,\ ^\forall r\in\mathcal{G},
 \end{align}
 where $\mathcal{G}$ consists of $m$ subsets of $\{1,\ldots,n\}$. 
 Suppose that there exists a pooling design $\bar{\mathcal{G}}$ satisfying 
 the constraint of \eqref{eqn:min-max-overlap} and the condition that 
 \begin{align}
  \begin{array}{l}
   \displaystyle
  i)\, \  |r\cap s|= \bar{\lambda}\ \text{or}\ \bar{\lambda}-1\  \text{for all\ $r,s\in
  \bar{\mathcal{G}},\ r\neq s$},\vspace*{1mm}\\ 
   \displaystyle   
  ii)  \  |\{r\in\bar{\mathcal{G}}~|~i\in r\}|=\bar{k}\ \text{or}\ \bar{k}-1
  \  \text{for all\ $i=1,\ldots,n$}, 
  \end{array}
  \label{eqn:opt-min-max-design}
 \end{align}
 where $\bar{\lambda}= \max_{r,s\in\bar{\mathcal{G}}:r\neq s}\ |r\cap s|$ and 
 $\bar{k}=\lceil{md/n}\rceil$. Then the pooling design $\bar{\mathcal{G}}$ is an optimal
 solution of \eqref{eqn:min-max-overlap}. 
\end{theorem}
\begin{proof}
 For a fixed pooling design $\mathcal{G}$ let $k_i$ be
 $k_i=|\{r\in\mathcal{G}~|~i\in r\}|$, that is $k_i$ stands for the number of pools
 including the clone $i$. Then we have the equality 
 \begin{align*}
  \sum_{i=1}^{n}k_i=md,\qquad
  \sum_{i=1}^{n}k_i(k_i-1)=
  \sum_{\substack{r,s\in\mathcal{G}\\ r\neq s}}|r\cap s|. 
 \end{align*}
 Since the mean value is less than or equal to the maximum value, we have
 \begin{align*}
  \max_{\substack{r,s\in\mathcal{G}\\ r\neq s}}|r\cap s|
  \geq
  \frac{1}{m(m-1)}\sum_{\substack{r,s\in\mathcal{G}\\ r\neq s}}|r\cap s|
  =\frac{\sum_{i=1}^{n}k_i^2-md}{m(m-1)}
 \end{align*}
 Some calculation leads that 
 the quadratic function $\sum_{i=1}^nk_k^2$ is minimized at 
 $(\bar{k}_1,\ldots,\bar{k}_n)=(\bar{k},\ldots,\bar{k},\bar{k}-1,\ldots,\bar{k}-1)$ 
 under the constraint that $\sum_{i=1}^{n}k_i=md$ for integers $k_1,\ldots,k_n$. 
 Thus for any pooling design $\mathcal{G}$, 
 the objective function in \eqref{eqn:min-max-overlap} is bounded below by 
 $(\sum_{i=1}^{n}\bar{k}_i^2-md)/m(m-1)$ which depends only on $n,m$ and $d$. 
 For the pooling design $\bar{\mathcal{G}}$ satisfying \eqref{eqn:opt-min-max-design}, 
 we have 
\begin{align*}
 \bar{\lambda}
 = \max_{\substack{r,s\in\bar{\mathcal{G}}\\ r\neq s}}|r\cap s|
 \geq \frac{\sum_{i=1}^{n}\bar{k}_i^2-md}{m(m-1)}
 >\bar{\lambda}-1. 
\end{align*}
The last inequality comes from the facts that $|r\cap s|$ is equal to $\bar{\lambda}$ or 
$\bar{\lambda}-1$ and that there exists a pair $r, s$ such that 
$|r\cap s|=\bar{\lambda}$. 
Thus $\bar{\mathcal{G}}$ is an optimal design, since $\bar{\mathcal{G}}$ attains the least
integer which is greater than or equal to the lower bound of the objective function. 
\end{proof}


The pooling design called balanced incomplete block design (BIBD) has the property 
such that in conditions i) and ii) of Theorem \ref{theorem:optimal_pooling_design}
equalities $|r\cap s|=\bar{\lambda}$
and $|\{r\in\mathcal{G}~|~i\in r\}|=\bar{k}$
always hold. 
According to Theorem \ref{theorem:optimal_pooling_design}, 
a BIBD is an optimal solution in the sense that it has the maximum
possible number of clones $n$ for given number of pools $m$ among the designs satisfying 
\eqref{eqn:min-max-overlap} if it exists for specified $n$, $m$ and $d$. 

A BIBD is often called a 2-design.
The existence condition and the construction method of BIBD's have been intensively investigated in the field of
combinatorics 
\citep{beth99:_desig_theor,colbourn07:_handb_of_combin_desig}. 
Among them, constructions based on 
finite fields and finite geometries are well investigated. Also many recursive constructions or composition
methods are developed. Tables of the existing BIBD's for small orders are listed in 
Chapter 2 of \citet{colbourn07:_handb_of_combin_desig}. 
The designs utilized in this paper 
are constructed based on Theorem 2 in 
\cite{wilson72:_cyclot_and_differ_famil_in}. See also Lemma 6.3 in 
\cite{beth99:_desig_theor} for details. 




\section{Numerical Experiments} 
\label{sec:Numerical_Experiments} 
The bias correction is examined in some numerical experiments. 
In the experiment, we specify the number of clones ($n$), the number of pools ($m$) and
the size of pool $|r|$, 
and then construct a pooling design $\mathcal{G}$ satisfying the condition 
$|r\cap s|=\lambda$ for any pair of pools $r,s\in\mathcal{G}$, where $\lambda$ is a
prespecified constant. Then, the group test is conducted by using the pooling design. 

In numerical experiments, the number of clones is set to $n=24, 1314$ or $1552$, 
and the pooling design is prepared based on the balanced incomplete block design. 
Table \ref{table:BIBD_pooling} illustrates the pooling design for each simulation. 
Basically, the same BIB designs are combined to make larger pooling design. 
In order to build the pooling design such that any pair of clones is not assigned
exactly the same pools, we applied randomization technique. 
The priori probability for each 
clone is defined as $p_i(X_i=1)=0.1$ for $n=24$ and $p_i(x_i)=0.002$ for $n=1314$ and $n=1552$. 
As shown in Table \ref{table:observation_prob} the conditional probability of the
observation, $p(S_r=s_r|Z_r=z_r)$, has been estimated by the experiments of an 
actual DNA library screening \citep{knill96:_inter_of_poolin_exper_using}, and thus we use
the probability in our algorithm. 

In the simulation, some positive clones are randomly chosen out of $n$ clones, and the
observations  
$s_r\in\{0,1,2,3\}, r\in \mathcal{G}$ are generated according to the defined
probability. 
The number of positive clones varies from one to
four. Then, we estimate the marginal posterior probability
$p_i(x_i|s)$. 
The estimated probability is compared to the true posterior probability 
computed by the Markov Chain Monte Carlo (MCMC) method
\cite{knill96:_inter_of_poolin_exper_using}. 
Table \ref{table:posterior-exam} shows the estimated result in the descending order of the
marginal posterior probability. In both methods almost the same clones are highly placed. 
Note that the MCMC method is computationally demanding. 
We use the MCMC method in order to obtain precise posterior probability 
which is used to assess the estimated (bias-corrected) posterior probability. 
In the numerical experiments, we use 
Concave-Convex Procedure (CCCP) algorithm \citep{yuille02:_cccp_algor_to_minim_bethe} to
compute the posterior probability instead of the conventional loopy BP algorithm. 
The CCCP has the same extremal solution as the loopy BP algorithm, though the CCCP may
have better convergence property. 
The computation time is shown in Table \ref{table:comp-time}. The CCCP is compared with
the MCMC method. Overall CCCP is efficient for large set of clones. 
We have confirmed that the computation time for bias correction is negligible. 

The bias-correction term $\frac{1}{2}\sum_{r\neq s:r,s\in\mathcal{G}}B_{rsi}$ is added to
the estimated posterior probability given by CCCP. 
The accuracy of the estimator is measured by the Kullback-Leibler (KL) divergence. 
Let $q_i(x_i)$ be the true posterior given by the MCMC method for $n=1314, 1552$. 
For $n=24$, the exact posterior probability is available. 
the discrepancy between $q_i$ and the estimated posterior $\bar{q}_i$ for $i=1,\ldots,n$
is measured by 
\begin{align*}
 \frac{1}{n}\sum_{i=1}^{n}\sum_{x_i\in\{0,1\}}
 q_i(x_i|s)\log\frac{q_i(x_i|s)}{\bar{q}_i(x_i|s)}. 
\end{align*}
In the numerical simulation we conducted the estimation 1000 times with different random
seed, and the KL-divergence is averaged over the repetition. 

Table \ref{table:clone-low}, Table \ref{table:clone-mid}, and Table \ref{table:clone-high}
show the results for each pooling design. 
The first column shows the number of positive clones out of $n$ clones, and the second and
third columns present the averaged KL-divergence for the estimator given by CCCP and its
bias-corrected variant, respectively. 
When $|r\cap s|$ is less than three, the bias correction works well to improve the accuracy
of the estimated posterior  
as shown in Table \ref{table:clone-low} and Table \ref{table:clone-high}. 
Table \ref{table:clone-mid} shows the result using the pooling design satisfying
$|r\cap s|=3$. 
In this case, the bias-correction does not necessarily improve the estimator. 
This result indicates that not only the dominant bias term 
$\frac{1}{2}\sum_{r\neq s}B_{rsi}$
but also the higher order term will be necessary to improve the estimator.  

In the simple experiments, the bias correction may be useful to improve the estimated
posterior when the pooling design $\mathcal{G}$ satisfies 
$|r\cap s|=2$ for $r,s\in \mathcal{G}$. 

\begin{table}
 \begin{center}
  \caption{
  Balanced in complete block (BIB) designs used in the simulation and the prior probability
  are shown. In our context the conventional notation $(v,r,b,k,\lambda)$
  for $\mathrm{BIBD}(v,r,b,k,\lambda)$  
  corresponds to 
  $(m,|r|,n,nm/|r|,|r\cap s|)$ for $r,s\in \mathcal{G},\,r\neq s$, 
  where $|r|$ and $|r\cap s|$ take a constant number. 
  The identical BIB designs are combined to make larger pooling design. 
  When the base design is 
  $\mathrm{BIBD}(v,r,b,k,\lambda)$ and the repetition is $t$, 
  the pooling design defined from 
  $\mathrm{BIBD}(v,r\cdot z,b\cdot t,k,\lambda\cdot t)$ 
  is constructed by combining the base design. 
  In order to build the pooling design such that any pair of clones is not assigned
  exactly the same pools, we applied randomization technique. } 
 \label{table:BIBD_pooling}
  \begin{tabular}{|l|c|c|c|}
    \hline
   \multicolumn{1}{|c|}{$\#$clones}   & base design  & repetition & prior: $p_i(X_i=1)$\\ \hline
   $n=24=12\times 2$     & $\mathrm{BIBD}(9,4,12,3,1)$    & 2 & 0.1    \\ 
   $n=1314=438\times 3$  & $\mathrm{BIBD}(73,24,438,4,1)$ & 3 & 0.002  \\ 
   $n=1552=776\times 2$  & $\mathrm{BIBD}(97,32,776,4,1)$ & 2 & 0.002  \\ \hline
  \end{tabular}
 \end{center}
 \begin{center}
  \caption{The conditional probability estimated by 
  the experiments of an actual DNA library screening
  \citep{knill96:_inter_of_poolin_exper_using}. }
 \label{table:observation_prob}
 \begin{tabular}{|c||c|} 
 \hline 
 $P(S_r=0|Z_r=0)=0.871$\  & $P(S_r=0|Z_r=1)=0.05$ \\
 $P(S_r=1|Z_r=0)=0.016$\  & $P(S_r=1|Z_r=1)=0.11$ \\
 $P(S_r=2|Z_r=0)=0.035$\  & $P(S_r=2|Z_r=1)=0.27$ \\
 $P(S_r=3|Z_r=0)=0.078$\  & $P(S_r=3|Z_r=1)=0.57$ \\  
\hline 
\end{tabular}
 \end{center}
 \begin{center}
  \caption{Estimated posterior probability in the preliminary experiments. 
  The estimated probability using loopy BP algorithm is compared to the true posterior
  computed by the Markov Chain Monte Carlo (MCMC) method
  \cite{knill96:_inter_of_poolin_exper_using}.  
  The result is shown in the descending order of the marginal posterior probability. 
  In both methods almost the same clones are highly placed. 
}
 \label{table:posterior-exam}
 \begin{tabular}{|c||c|c||c|c|} \hline
     &  \multicolumn{2}{c||}{loopy BP}  & \multicolumn{2}{c|}{MCMC} \\\hline
rank & clone id &  posterior            & clone id & posterior \\\hline
  1  &  336     & 0.8393                & 336      & 0.8345 \\
  2  &  768     & 0.0615                & 768      & 0.0628 \\
  3  &  125     & 0.0574                & 125      & 0.0608 \\
  4  &  764     & 0.0419                &  81      & 0.0400 \\
  5  &  81      & 0.0409                & 764      & 0.0382 \\\hline
\end{tabular}
\end{center}
  \begin{center}
  \caption{
   The computation time (second) is shown. The CCCP is compared with
   the MCMC method. Overall CCCP is efficient for the large set of clones. 
   We have confirmed that the computation time for bias correction is negligible. } 
  \label{table:comp-time}
   \begin{tabular}{|c|c|c|} \hline 
 $n$  & CCCP  &  MCMC   \\ \hline
981   & 0.22  &  1.91   \\
1298  & 0.27  &  2.49   \\
3088  & 0.81  &  8.49   \\
6371  & 1.68  &  17.60  \\
10121 & 3.33  &  27.73  \\
30050 &11.09  &  81.80  \\ \hline
  \end{tabular}
 \end{center}
\end{table}

\begin{table}
 \begin{center}
  \caption{
  The numerical results for pooling design such that 
  $n=24,\ m=9,\ |r|=8,\ |r\cap s|=2$ are shown. 
  The prior probability is set to $p_i(X_i=1)=0.1$ for all $i=1,\ldots,n$. 
  The first column shows the number of positive clones out of $n$ clones, and the second and
  third column presents the averaged KL-divergence for the CCCP and its bias-corrected
  variant from the posterior given by the MCMC method, respectively. }
  \label{table:clone-low}
  \begin{tabular}{|c|c|c|} \hline
   \#   positive &  {\qquad\  CCCP\qquad\  } &     bias-corrected CCCP\\\hline
   1 &  11.67e-04 & 6.67e-04 \\
   2 &  10.57e-04 & 6.01e-04 \\
   3 &  7.020e-04 & 4.26e-04 \\
   4 &  4.160e-04 & 2.76e-04 \\\hline
  \end{tabular}
  \caption{
  The numerical results for pooling design such that 
  $n=1314,\ m=73,\ |r|=72,\ |r\cap s|=3$ are shown. 
  The prior probability is set to $p_i(X_i=1)=0.002$ for all $i=1,\ldots,n$. 
  }
  \label{table:clone-mid}
  \begin{tabular}{|c|c|c|} \hline
   \# positive &  {\qquad\ CCCP\qquad\ } &     bias-corrected CCCP\\\hline
   1 & 3.80\,e-05 & 2.40\,e-05\\ 
   2 & 1.80\,e-05 & 1.80\,e-05\\ 
   3 & 10.1\,e-05 & 14.3\,e-05\\ 
   4 & 4.80\,e-05 & 5.20\,e-05\\ \hline
  \end{tabular}
  \caption{
  The numerical results for pooling design such that 
  $n=1552,\ m=97,\ |r|=64,\ |r\cap s|=2$ are shown. 
  The prior probability is set to $p_i(X_i=1)=0.002$ for all $i=1,\ldots,n$. 
  }
  \label{table:clone-high}
  \begin{tabular}{|c|c|c|} \hline
   \#   positive &  {\qquad\ CCCP\qquad\ } &     bias-corrected CCCP\\\hline
   1 &  0.90\,e-05 & 0.90\,e-05 \\ 
   2 &  1.70\,e-05 & 1.50\,e-05 \\ 
   3 &  3.30\,e-05 & 1.90\,e-05 \\ 
   4 &  2.80\,e-05 & 2.60\,e-05 \\ \hline 
  \end{tabular}
 \end{center}
\end{table}

\section{Concluding Remarks}
\label{sec:ConcludingRemarks}
For the pooling design we have proposed the bias corrected estimator of the marginal
posterior probability based on the result of 
\cite{ikeda04:_infor_geomet_of_turbo_and,ikeda04:_stoch_reason_free_energ_and_infor_geomet}. 
We analyzed an upper bound of the bias term and showed that BIB design will make the bias
small comparing to other pooling designs. 
In numerical experiments, the bias correction works well to improve the marginal 
posterior, even when $|r\cap s|=2$ holds for the pooling design $\mathcal{G}$. 
We confirmed that the correction of the dominant bias term does not necessarily improve 
the estimator, when the pooling design satisfies $|r\cap s|=3$. 
Investigating higher order bias correction will be an important future
work. 

\appendix

\section{Calculation of $B_{rsi}$}  
\label{appendix:B_calc} 
Theorem \ref{theorem:packing-optimal} is obtained as a direct conclusion of Theorem 7 in
\cite{ikeda04:_infor_geomet_of_turbo_and}. 
Here, we compute $B_{rsi}$ to show its explicit form, and verify that $B_{rsi}=0$ for
$|r\cap s|\leq 1$. 

As shown in \eqref{eqn:c_function_canonical_form} and \eqref{eqn:cr-other-expression}, any 
function $c_r(x)$ over $\{0,1\}^n$ 
depending on only $x_i, i\in r$ is represented by the linear sum of the functions having
the form of $\prod_{i\in r}(x_i-a_i)$, where 
$a_i\in [0,1]$.  
Moreover, the bias term $B_{rsi}$ is bilinear in $c_r(x)-\bar{c}_r$ and
$c_s(x)-\bar{c}_s$. Therefore, it is enough to consider the case that $c_r$ and $c_s$ are
given as $c_r(x)=\prod_{i\in r}(x_i-a_{ri})$ and $c_s(x)=\prod_{i\in s}(x_i-a_{si})$ 
for $a_{ri}, a_{si}\in [0,1]$. 

Let us define $e_r(r')$ for the subset $r'\subset r$ by 
\begin{align*}
 e_r(r')=\prod_{i\in r'}(\bar{x}_i-a_{ri}). 
\end{align*}
Let $E[\,\cdot\,]$ be the expectation by the probability $\bar{q}(x;\theta_0)$, 
then we have $E[c_r]=e_r(r)$. 
Building blocks for the calculation of the bias term are given as follows. 
The matrix $g_{ij}, g_{ir}$ and $\widetilde{g}_{ir}$ are given as
\begin{align*}
 g_{ii}&=\bar{x}_i(1-\bar{x}_i),\quad
 g_{ir}=
 \begin{cases}
  0& i\not\in r,\\
  \bar{x}_i(1-\bar{x}_i) e_r(r\setminus\{i\}) & i\in r, 
 \end{cases}\\
 \therefore\ \ \ 
 \widetilde{g}_{ir}&=\frac{g_{ir}}{g_{ii}}=
 \begin{cases}
  0& i\not\in r,\\
  e_r(r\setminus\{i\}) & i\in r. 
 \end{cases}
\end{align*}
The third tensor $T_{ijk}$ is computed as follows:
\begin{align*}
 T_{ijk}&=\bar{x}_i(1-\bar{x}_i)(1-2\bar{x}_i)\delta_{ij}\delta_{ik},\\
 \therefore\ \ 
 \sum_{jk}T_{ijk}\widetilde{g}_{jr}\widetilde{g}_{ks}&=
 \begin{cases}
 \bar{x}_i(1-\bar{x}_i)(1-2\bar{x}_i)e_r(r\setminus\{i\})e_s(s\setminus\{i\}) 
  & i\in r\cap s,\\
  0  & \text{otherwise}
 \end{cases}
\end{align*}

For $T_{ijr}$ we see that $T_{ijr}=0$ when $i\not\in r$ or $j\not\in r$ holds. 
Then, we compute $T_{ijr}$ for $i,j\in r$. When $i=j\in r$ holds, we have  
\begin{align*}
 T_{iir}
 &=E[(x_i-\bar{x}_i)^2(c_r(x)-\bar{c}_r)]\\
 &=E[(x_i-\bar{x}_i)^2(x_i-a_{ri})]e_r(r\setminus\{i\})-\bar{x}_i(1-\bar{x}_i)e_r(r)\\
 &= \bar{x}_i(1-\bar{x}_i)(1-2\bar{x}_i)e_r(r\setminus\{i\}). 
\end{align*}
In the same way, we obtain 
\begin{align*}
 T_{ijr}=
 \begin{cases}
 \bar{x}_i(1-\bar{x}_i)(1-2\bar{x}_i)  e_r(r\setminus\{i\})  & i=j\in r\\
  \bar{x}_i(1-\bar{x}_i)\bar{x}_j(1-\bar{x}_j)e_r(r\setminus\{i,j\}) & i,j\in r, i\neq j
 \end{cases}
\end{align*}
Note that
\begin{align*}
 \sum_{j}(T_{ijr}\widetilde{g}_{js}+T_{ijs}\widetilde{g}_{jr})
 &=\sum_{j\in r\cap s}(T_{ijr}\widetilde{g}_{js}+T_{ijs}\widetilde{g}_{jr})
\end{align*}
holds. For $i\not\in r\cup s$, the equality $T_{ijs}=T_{ijr}=0$ holds, and for 
$r\cap s=\emptyset$ we have $T_{ijs}\widetilde{g}_{jr}=0$. Thus, 
\begin{align*}
i\not\in r\cup s\ \ \text{or}\ \ r\cap s=\emptyset\  \Longrightarrow \ 
\sum_{j}(T_{ijr}\widetilde{g}_{js}+T_{ijs}\widetilde{g}_{jr})=0
\end{align*}
holds. 
Then for $r\cap s\neq \emptyset$ we consider other two cases: $i\in r\setminus s$
(or $i\in s\setminus r$) and $i\in r\cap s$. Some calculation leads to 
\begin{align*}
  i\in r\backslash s:&\ \ 
  \sum_{j\in r\cap s}(T_{ijr}\widetilde{g}_{js}+T_{ijs}\widetilde{g}_{jr})
  ~=~
  \bar{x}_i(1-\bar{x}_i)\sum_{j\in r\cap s}\bar{x}_j(1-\bar{x}_j) e_r(r\backslash\{i,j\}) e_s(s\backslash \{j\}),\\
  i\in r\cap s:&\ \ 
 \sum_{j\in r\cap s}(T_{ijr}\widetilde{g}_{js}+T_{ijs}\widetilde{g}_{jr})\\
 &=
 2\bar{x}_i(1-\bar{x}_i)(1-2\bar{x}_i) e_r(r\backslash\{i\})  e_s(s\backslash\{i\}) \\
 &\phantom{=} +\bar{x}_i(1-\bar{x}_i) 
 \!\!\!\sum_{j\in r\cap s\backslash\{i\}} \bar{x}_j(1-\bar{x}_j)
  \big[e_r(r\backslash\{i,j\})e_s(s\backslash\{j\})+e_s(s\backslash\{i,j\})e_r(r\backslash\{j\})\big]. 
\end{align*}

For $T_{irs}$, we see that 
\begin{align*}
i\not\in r\cup s \ \ \text{or}\ \ r\cap s=\emptyset \ \Longrightarrow \ T_{irs}=0 
\end{align*}
holds. Under the condition that $r\cap s\neq\emptyset$, we consider other two cases: 
$i\in r\setminus s$ and $i\in r\cap s$, 
\begin{align*}
 i\in r\backslash s:&\\
 T_{irs}
 &=
 \bar{x}_i(1-\bar{x}_i) \bigg[
 e_r(r\backslash (s\cup \{i\}))e_s(s\backslash r)
 \prod_{j\in r\cap s}\{\bar{x}_j(1-\bar{x}_j)+(\bar{x}_j-a_{rj})(\bar{x}_j-a_{sj})\}\\
 &\qquad 
 -e_r(r\backslash\{i\})e_s(s)\bigg],\\
 i\in r\cap s:&\\
 T_{irs}
 &=
 \bar{x}_i(1-\bar{x}_i)
 \bigg[ (1-a_{ri}-a_{si}) e_r(r\backslash s)e_s(s\backslash r)
 \!\!\!\prod_{j\in r\cap s\backslash\{i\}}\!\!
 \{\bar{x}_j(1-\bar{x}_j)+(\bar{x}_j-a_{rj})(\bar{x}_j-a_{sj})\}\\
 &\qquad -e_r(r)e_s(s\backslash\{i\}) -e_r(r\backslash\{i\})e_s(s)
\bigg]. 
\end{align*}

We show $B_{rsi}$ below. Remember that 
\begin{align*}
 B_{rsi}
 =-T_{irs}-\sum_{j,k=1}^n T_{ijk}\widetilde{g}_{jr}\widetilde{g}_{ks}
 +\sum_{j=1}^{n}(T_{ijr}\widetilde{g}_{js}+T_{ijs}\widetilde{g}_{jr}). 
\end{align*}
Then, we have
\begin{align*}
 r\cap s=\emptyset\ \ \text{or}\ \ i\not\in r\cap s\ \ \Longrightarrow \ \ B_{rsi}=0, 
\end{align*}
since all terms vanish. 
To compute other cases, we use the formula, 
\begin{align*}
 &\phantom{=}
 \prod_{j\in r\cap s}\{\bar{x}_j(1-\bar{x}_j)+(\bar{x}_j-a_{rj})(\bar{x}_j-a_{sj})\}\\
 &=
 \sum_{k=0}^{|r\cap s|}
 \sum_{\substack{v\subset r\cap s\\ |v|=k}}
 e_r((r\cap s)\setminus v) \cdot e_s((r\cap s)\setminus v)
 \cdot\prod_{j\in v}\bar{x}_j(1-\bar{x}_j), 
\end{align*}
and so forth. 
Then $B_{rsi}$ is represented as follows, 
\begin{align*}
 i\in r\backslash s:&\ \ 
 B_{rsi}=
 -\bar{x}_i(1-\bar{x}_i)
 \sum_{k=2}^{|r\cap s|}\sum_{\substack{v\subset r\cap s\\ |v|=k}}\ 
 e_r(r\backslash (v\cup\{i\}))e_s(s\setminus v)
 \prod_{j\in v}\bar{x}_j(1-\bar{x}_j),\\
 i\in r\cap s:&\ \ 
 B_{rsi}=
 \bar{x}_i(1-\bar{x}_i) 
 \bigg[
 (2\bar{x}_i-1)\sum_{k\in r\cap s\backslash\{i\}}\bar{x}_k(1-\bar{x}_k) 
 e_r(r\backslash\{i,k\})e_s(s\backslash\{i,k\})\\
 &\quad 
 -(1-a_{ri}-a_{si}) \sum_{k=2}^{|r\cap s|-1}
 \sum_{\substack{v\subset r\cap s\backslash\{i\}\\ |v|=k}}
 \!\!\!\!
 e_r(r\backslash(\{i\}\cup v))e_s(s\backslash(\{i\}\cup v))
\prod_{j\in v}\bar{x}_j(1-\bar{x}_j)\bigg]. 
\end{align*}
Paying attention to the summation, we see that $B_{rsi}$ vanishes for $|r\cap s|\leq1$.

\section{Upper bound of $|B_{rsi}|$}
\label{appendix:Upperbound_Brsi}
First we derive an upper bound of $|B_{rsi}|$ for 
$c_r(x)=\prod_{i\in r}(x_i-a_{ri})$ and $c_s(x)=\prod_{i\in s}(x_i-a_{si})$. 
Suppose that there exists $\delta$ such that 
$\displaystyle 0< |\bar{x}_i-a_i|\leq \delta<1$ holds for all $i=1,\ldots,n$ and 
all $a_i$ appearing in $c_r$ and $c_s$. 
Then, we obtain an upper bound of $ |B_{rsi}|$. 
We use $\bar{x}_i(1-\bar{x}_i)\leq 1/4$. 
For $i\in r\backslash s$ we have
\begin{align*}
 |B_{rsi}| \leq 
 \frac{1}{4}\sum_{k=2}^{|r\cap s|}\sum_{\substack{v\subset r\cap s\\ |v|=k}}
 \delta^{|r|-1-k}\delta^{|s|-k}\left(\frac{1}{4}\right)^k
 \leq \frac{\delta^{|r|+|s|-1}}{4}\bigg(1+\frac{1}{4\delta^2}\bigg)^{|r\cap s|}
\end{align*}
and in the same way for $i\in r\cap s$ we obtain
\begin{align*}
 |B_{rsi}|
 \leq 
 \frac{\delta^{|r|+|s|-2}}{4}
 \bigg[
 \frac{|r\cap s|-1}{4\delta^2}
 +\bigg(1+\frac{1}{4\delta^2}\bigg)^{|r\cap s|-1}
 \bigg]
 \leq  
 \frac{\delta^{|r|+|s|-2}}{2}
 \bigg(1+\frac{1}{4\delta^2}\bigg)^{|r\cap s|-1}. 
\end{align*}
Therefore, for any case we have
\begin{align*}
 |B_{rsi}|\leq 
 \frac{\delta^{|r|+|s|-2}}{2}
 \bigg(1+\frac{1}{4\delta^2}\bigg)^{|r\cap s|}. 
\end{align*}
Next we suppose that
\begin{align*}
 c_r(x)&=h_r+\sum_{\ell}b_{r\ell}\prod_{i\in r}(x_i-a_{r\ell i}),\\
 c_s(x)&=h_s+\sum_{\ell'}b_{s\ell'}\prod_{i'\in s}(x_{i'}-a_{s\ell' i'}). 
\end{align*}
Since $B_{rsi}$ is bilinear in $c_r(x)-\bar{c}_r$ and $c_s(x)-\bar{c}_s$, 
we have 
\begin{align*}
 |B_{rsi}|
 \leq 
 \sum_{\ell,\ell'}\big|b_{r\ell}b_{s\ell'}\big|\cdot
 \frac{\delta^{|r|+|s|-2}}{2}
 \bigg(1+\frac{1}{4\delta^2}\bigg)^{|r\cap s|}
 \leq 
 C\cdot
 \frac{\delta^{|r|+|s|-2}}{2}
 \bigg(1+\frac{1}{4\delta^2}\bigg)^{|r\cap s|}. 
\end{align*}

\bibliographystyle{jstpip}
\bibliography{dna-lib}

\begin{thebibliography}{17}
\providecommand{\natexlab}[1]{#1}
\expandafter\ifx\csname urlstyle\endcsname\relax
  \providecommand{\doi}[1]{doi:\discretionary{}{}{}#1}\else
  \providecommand{\doi}{doi:\discretionary{}{}{}\begingroup
  \urlstyle{rm}\Url}\fi

\bibitem[{Beth et~al.(1999)Beth, Jungnickel, and Lenz}]{beth99:_desig_theor}
Beth T., Jungnickel D., Lenz H., 1999.
\newblock \emph{Design Theory}.
\newblock Cambridge Univ. Press.

\bibitem[{Bruno et~al.(1995)Bruno, Knill, Balding, Bruce, Doggett, Sawhill,
  Stallings, Whittaker, and
  Torney}]{bruno95:_effic_poolin_desig_for_librar_screen}
Bruno W., Knill E., Balding D., Bruce D., Doggett N., Sawhill W., Stallings R.,
  Whittaker C., Torney D., 1995.
\newblock Efficient pooling designs for library screening.
\newblock \emph{Genomics}, 26(1), 21--30.

\bibitem[{Colbourn and Dinitz(2007)}]{colbourn07:_handb_of_combin_desig}
Colbourn C.J., Dinitz J.H. (eds.), 2007.
\newblock \emph{Handbook of combinatorial designs}.
\newblock CRC Press, Boca Raton, FL, second edition.

\bibitem[{Cowell et~al.(2007)Cowell, Dawid, Lauritzen, and
  Spiegelhalter}]{cowell07:_probab_networ_and_exper_system}
Cowell R.G., Dawid A.P., Lauritzen S.L., Spiegelhalter D.J., 2007.
\newblock \emph{Probabilistic Networks and Expert Systems: Exact Computational
  Methods for Bayesian Networks}.
\newblock Springer Publishing Company, Incorporated.

\bibitem[{Du and Hwang(1999)}]{d-z.99:_combin_group_testin_and_its_applic_ed}
Du D.Z., Hwang F.K., 1999.
\newblock \emph{Combinatorial Group Testing and Its Applications, 2nd ed.}
\newblock World Scientific.

\bibitem[{Ikeda et~al.(2004{\natexlab{a}})Ikeda, Tanaka, and
  Amari}]{ikeda04:_infor_geomet_of_turbo_and}
Ikeda S., Tanaka T., Amari S., 2004{\natexlab{a}}.
\newblock Information geometry of turbo and low-density parity-check codes.
\newblock \emph{IEEE Transactions on Information Theory}, 50(6), 1097--1114.

\bibitem[{Ikeda et~al.(2004{\natexlab{b}})Ikeda, Tanaka, and
  Amari}]{ikeda04:_stoch_reason_free_energ_and_infor_geomet}
Ikeda S., Tanaka T., Amari S., 2004{\natexlab{b}}.
\newblock Stochastic reasoning, free energy, and information geometry.
\newblock \emph{Neural Comput.}, 16(9), 1779--1810.

\bibitem[{Knill et~al.(1996)Knill, Schliep, and
  Torney}]{knill96:_inter_of_poolin_exper_using}
Knill E., Schliep A., Torney D., 1996.
\newblock Interpretation of pooling experiments using the markov chain monte
  carlo method.
\newblock \emph{J. Computational Biology}, 3(3), 395--406.

\bibitem[{MacKay(1999)}]{mackay99:_good_error_correc_codes_based}
MacKay D.J.C., 1999.
\newblock Good error correcting codes based on very sparse matrices.
\newblock \emph{IEEE Transactions on Information Theory}, 45(2), 399--431.

\bibitem[{M\'{e}zard and
  Toninelli(2007)}]{mezard07:_group_testin_with_random_pools}
M\'{e}zard M., Toninelli C., 2007.
\newblock Group testing with random pools: Optimal two-stage algorithms.
\newblock \emph{arXiv:0706.3104}.

\bibitem[{Ngo and Du(2000)}]{ngo00:_survey_combin_group_testin_algor}
Ngo H., Du D., 2000.
\newblock A survey on combinatorial group testing algorithms with applications
  to dna library screening.
\newblock \emph{Discrete Math. Problems with Medical Applications, DIMACS Ser.
  Discrete Math. and Theoretical Computer Science}, 55, 171--182.

\bibitem[{Pearl(1988)}]{pearl88:_probab_reason_in_intel_system}
Pearl J., 1988.
\newblock \emph{Probabilistic Reasoning in Intelligent Systems: Networks of
  Plausible Inference}.
\newblock Morgan Kaufmann Publishers Inc., San Francisco, CA, USA.
\newblock ISBN 1558604790.

\bibitem[{Richardson et~al.(2001)Richardson, Shokrollahi, and
  Urbanke}]{richardson01:_desig_of_capac_approac_irreg}
Richardson T.J., Shokrollahi M.A., Urbanke R.L., 2001.
\newblock Design of capacity-approaching irregular low-density parity-check
  codes.
\newblock \emph{IEEE Transactions on Information Theory}, 47(2), 619--637.

\bibitem[{Uehara and Jimbo(2009)}]{ueharaar:_posit_detec_code_and_its}
Uehara H., Jimbo M., 2009.
\newblock A positive detecting code and its decoding algorithm for dna library
  screening.
\newblock \emph{IEEE/ACM Trans. Comput. Biol. Bioinformatics}, 6(4), 652--666.

\bibitem[{Wilson(1972)}]{wilson72:_cyclot_and_differ_famil_in}
Wilson R.M., 1972.
\newblock Cyclotomy and difference families in elementary abelian groups.
\newblock \emph{J. Number Theory}, 4, 17--47.

\bibitem[{Wu et~al.(2004)Wu, Li, Huang, and
  D.Z.Du}]{wu04:_molec_biolog_and_poolin_desig}
Wu W., Li Y., Huang C., D.Z.Du, 2004.
\newblock Molecular biology and pooling design.
\newblock \emph{Proc. Workshop Data Mining in Biomedicine (DMB '04)}.

\bibitem[{Yuille(2002)}]{yuille02:_cccp_algor_to_minim_bethe}
Yuille A.L., 2002.
\newblock {CCCP} algorithms to minimize the bethe and kikuchi free energies:
  convergent alternatives to belief propagation.
\newblock \emph{Neural Comput.}, 14(7), 1691--1722.
\newblock ISSN 0899-7667.

\end{thebibliography}

\end{document}